\documentclass{article}
\usepackage{fullpage}
\usepackage{graphicx,multirow}
\usepackage{amsmath,amssymb,amsthm}
\usepackage{authblk}
\usepackage{color,tikz}
\usepackage[subrefformat=parens,labelformat=parens]{subfig}
\usetikzlibrary{patterns}
\usetikzlibrary{decorations.pathmorphing}
\usepackage{pgfgantt}
\newtheorem{theorem}{Theorem}
\newtheorem{lemma}{Lemma}
\newtheorem{corollary}{Corollary}

\newcommand{\opt}{\textsc{Opt}}

\newcommand{\argmax}{\rm{argmax}}

\definecolor{grey}{gray}{0.75}

\newcommand{\email}[1]{{\tt \small #1}}
\begin{document}

\title{Split scheduling with uniform setup times}

\author[1]{Frans Schalekamp}
\author[2,3]{Ren{\'e} Sitters}
\author[2]{Suzanne van der Ster}
\author[2,3]{Leen~Stougie}
\author[4]{V{\'i}ctor Verdugo}
\author[1]{Anke van Zuylen}
\affil[1]{College of William \& Mary, Department of Mathematics, Williamsburg, VA 23185, USA\\            \email{\{frans, anke\}@wm.edu} }
\affil[2]{Department of Operations Research, VU University Amsterdam, The Netherlands\\
           \email{\{r.a.sitters, suzanne.vander.ster, l.stougie\}@vu.nl} }
\affil[3]{CWI, Amsterdam, The Netherlands\\
           \email{\{sitters, stougie\}@cwi.nl}\thanks{This research was partially supported by Tinbergen Institute.}}
\affil[4]{Department of Industrial Engineering, University of Chile, Santiago, Chile\\
           \email{vverdugo@dim.uchile.cl}\thanks{Research supported by EU-IRSES grant EUSACOU.
           }}

\maketitle

\begin{abstract}
\noindent We study a scheduling problem in which jobs may be split into parts, where the parts of a split job may be processed simultaneously on more than one machine. Each part of a job requires a setup time, however, on the  machine where the job part is processed. During setup a machine cannot process or set up any other job. We concentrate on the basic case in which setup times are job-, machine-, and sequence-independent. Problems of this kind were encountered when modelling practical problems in planning disaster relief operations. Our main algorithmic result is a polynomial-time algorithm for minimising total completion time on two parallel identical machines. We argue why the same problem with three machines is not an easy extension of the two-machine case, leaving the complexity of this case as a tantalising open problem. We give a constant-factor approximation algorithm for the general case with any number of machines and a polynomial-time approximation scheme for a fixed number of machines. For the version with objective minimising {\em weighted} total completion time we prove NP-hardness. Finally, we conclude with an overview of the state of the art for other split scheduling problems with job-, machine-, and sequence-independent setup times.
\end{abstract}

\section{Introduction}\label{sec:intro}
We consider a scheduling problem with {\em setup times} and {\em job splitting}.
Given a set of identical parallel machines and a set of jobs with processing times, the goal of the scheduling problem is to schedule the jobs on the machines such that a given objective, for example the makespan or the sum of completion times, is minimised. With ordinary preemption, feasible schedules do not allow multiple machines to work on the same job simultaneously. In {\em job splitting}, this constraint is dropped. Without setup times, allowing job splitting makes many scheduling problems trivial: both for minimising make-span and for minimising total (weighted) completion time, an optimal schedule is obtained by splitting the processing time of each job equally over all machines, and processing the jobs in arbitrary order on each machine in case of makespan, and in {\em (weighted) shortest processing time first} ((W)SPT) order in case of total (weighted) completion time. See Xing and Zhang~\cite{XingZ00} for an overview of several classical scheduling problems which become polynomially solvable if job splitting is allowed.

In the presence of release times, minimising total completion time with ordinary preemption is NP-hard \cite{DuLeungYoung}, whereas it is easy to see that if we allow job splitting, then splitting all jobs equally over all machines and applying the {\em shortest remaining processing time first} (SRPT) rule gives an optimal schedule.

Triviality disappears when setup times are present, i.e., when each machine requires a setup time before it can start processing the next job (part).
During setup, a machine cannot process any job nor can it set up the processing of any other job (part). Problems for which the setup times are allowed to be sequence-dependent are usually NP-hard, as such problems tend to exhibit routinglike features.
For example, the Hamiltonian path problem in a graph can be reduced to the problem of minimising the makespan on a single machine, where each job corresponds to a node in the graph, the processing times are 1, and the setup time between job $i$ and $j$ is 0 if the graph contains an edge between $i$ and $j$, and 1 otherwise.
However, as we will see, adding setup times leads to challenging algorithmic problems, already if the setup times are assumed to be job-, machine-, and sequence-independent. 

We encountered such problems in studying disaster relief operations. For example in modelling flood relief operations, the machines are pumps and the jobs are locations to be drained. Or in the case of earthquake relief operations, the machines are teams of relief workers and the jobs are locations to be cleared. The setup is the time required to install the team on the new location.
Although, in principle, these setup times consist partly of travel time, which are sequence-dependent, the travel time is negligible compared to the time required to equip the teams with instructions and tools for the new location. Hence, considering the setup times as being location- and sequence-independent was in this case an acceptable approximation of reality.

In this paper we concentrate on a basic scheduling problem and consider the variation where we allow job splitting with setup times that are job-,
machine-, and sequence-independent, to which we will refer here as {\em uniform setup times}; i.e., we assume a \emph{uniform} setup time $s$.  There exists little literature on this type of scheduling problems. The problem of minimising makespan on parallel identical machines is in the standard scheduling notation of Graham et al. \cite{graham1979} denoted as $P||C_{\max}$ (see Section \ref{sec:epilogue} for an instruction on this notation). This problem $P||C_{\max}$, but then with job splitting and setup times that are job-dependent, but sequence- and machine-independent, is considered by Xing and Zhang~\cite{XingZ00}, and Chen, Ye and Zhang~\cite{chen2006lot}. Chen et al.~\cite{chen2006lot} mention that this problem is NP-hard in the strong sense, and only weakly NP-hard if the number of machines is assumed constant. Straightforward reductions from the {\sc 3-Partition} and {\sc Subset Sum} problem show that these hardness results continue to hold if setup times are uniform. Chen et al. provide a 5/3-approximation algorithm for this problem and an FPTAS for the case of a fixed number of machines. A PTAS for the version of $P||C_{\max}$ with preemption and job-dependent, but sequence- and machine-independent setup times was given by Schuurman and Woeginger~\cite{schuurman1999preemptive}. It remains open whether a PTAS exists with job splitting rather than preemption, even if the setup times are uniform. See~\cite{LiuC04} and~\cite{PottsvW92} for a more extensive literature on problems with preemption and setup times.

Our problem is related to scheduling problems with malleable tasks. A malleable task may be scheduled on multiple machines, and a function $f_j(k)$ is given that denotes the processing speed if $j$ is processed on $k$ machines. If $k$ machines process task $j$ for $L$ time, then $f_j(k)L$ units of task $j$ are completed.  What we call job splitting is referred to as malleable tasks with linear speedups, i.e., the processing time required on $k$ machines is $1/k$ times the processing time required on a single machine. We remark that job splitting with setup times is not a special case of scheduling malleable tasks, because of the discontinuity caused by the setup times. We refer the reader to Drozdowski~\cite{Drozdowski09} for an extensive overview of the literature on scheduling malleable tasks.

The main algorithmic result of our paper considers the job splitting variant of the problem of minimising the sum of completion times on identical machines, with uniform setup times:  given a set of $m$ identical machines, $n$ jobs with processing times $p_1, \ldots, p_n$, and a setup time $s$,
the objective is to schedule the jobs on the machines to minimise total completion time ($\sum C_j$) (where the chosen objective is inspired by the disaster relief application).
The version of this problem with ordinary preemption and fixed setup time $s$ is solved by the Shortest Processing Time first rule (SPT); the option of preemption is not used by the optimum. However, the situation is much less straightforward for job splitting. If $s$ is very large, then an optimal schedule minimises the contribution of the setup times to the objective, and a job will only be split over several machines if no other job is scheduled after the job on these machines. It is not hard to see that the jobs that are not split are scheduled in SPT order.
If $s$ is very small (say 0), then each job is split over all machines and the jobs are scheduled in SPT order. However, for other values of $s$, it appears to be a non-trivial problem to decide how to schedule the jobs, as splitting a job over multiple machines decreases the completion time of the job itself, but it increases the total load on the machines, and hence the completion times of later jobs.

Consider the following instance as an example. There are 3 machines and 6 jobs, numbered $1, 2, \ldots, 6$, with processing times $1, 2, 3, 5,  11, 12$, respectively, and setting up a machine takes 1 time unit. One could consider filling up a schedule in round-robin style, assigning the jobs to machine 1, 2, 3, 1, 2, 3, respectively. This schedule is given in the Gantt chart in Figure~\subref*{fig:example1}. The schedule has objective value $49$.

\begin{figure*}[h!t]
\begin{center}
\subfloat[]{\label{fig:example1}\noindent\resizebox{.7\textwidth}{!}{
\begin{ganttchart}[hgrid, inline, bar height = 1, bar top shift = 0, x unit = .7cm, y unit chart= .4cm]{16}
\input{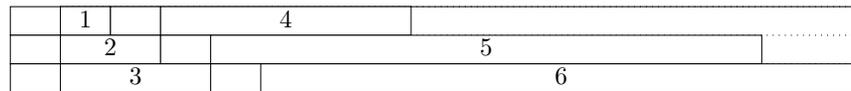}
\end{ganttchart}
}}\\
\subfloat[]{\label{fig:example2}\noindent\resizebox{.7\textwidth}{!}{
\begin{ganttchart}[hgrid, inline, bar height = 1, bar top shift = 0, x unit = .7cm, y unit chart= .4cm]{16}
\input{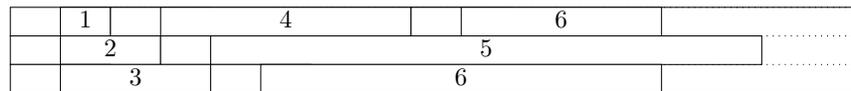}
\end{ganttchart}
}}\\
\subfloat[]{\label{fig:example3}\noindent\resizebox{.7\textwidth}{!}{
\begin{ganttchart}[hgrid, inline, bar height = 1, bar top shift = 0, x unit = .7cm, y unit chart= .4cm]{16}
\input{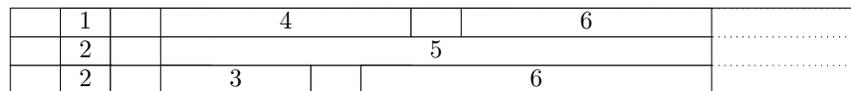}
\end{ganttchart}
}}\\
\subfloat[]{\label{fig:example4}\noindent\resizebox{.7\textwidth}{!}{
\begin{ganttchart}[hgrid, inline, bar height = 1, bar top shift = 0, x unit = .7cm, y unit chart= .4cm]{16}
\input{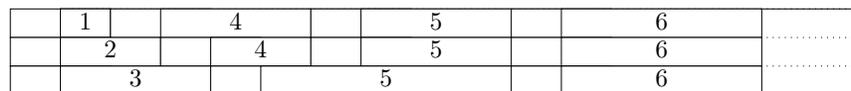}
\end{ganttchart}
}}
\caption{Gantt charts depicting the schedules for the instance described in Section~\ref{sec:P2}. The grey blocks indicate the setup times, the numbered blocks are scheduled job parts. Each row of blocks gives the schedule for a machine.}
\end{center}
\end{figure*}

By splitting job 6 over machines 1 and 3, instead of processing it on machine 3 only, we can lower the completion time of job 6, and this improves the objective value since there are no jobs scheduled after job 6. In fact, to get the best improvement in objective value, we make sure that both job parts of job 6 finish at the same time, see Figure~\subref*{fig:example2}. The objective value is of the schedule is $45$.

Splitting jobs early in the schedule, may increase the objective value, as (many) later jobs may experience delays. For example, if we choose to split job 2 over machines 2 and 3, we will cause delays for jobs 3 and 6, while improving the completion times of jobs 2 and 5. If we require that job parts of the same job end at the same time, we get the schedule pictured in Figure~\subref*{fig:example3} with objective value $46$. Figure~\subref*{fig:example4} depicts the optimal schedule with objective value $40$.

This illustrates the inherent trade-off in this problem mentioned earlier: splitting jobs will decrease the completion times of some jobs, but it also may increase the completion times of other jobs.

In Section~\ref{sec:prop2} we present a polynomial-time algorithm for the case in which there are two machines. The algorithm is based on a careful analysis of the structure of optimal solutions to this problem. Structures of optimal solutions that hold under any number of machines are presented in a preliminary section. Though a first guess might be that the problem would remain easy on any fixed number of machines, we will show by some examples in Section~\ref{sec:3machines} that nice properties, which make the algorithm work for the 2-machine case, fail to hold for three machines already. The authors are split between thinking that we have encountered another instance of Lawler's ``mystical power of twoness''~\cite{lenstra88}, a phrase signifying the surprisingly common occurance that problems are easy when a problem parameter (here the number of machines) is two, but NP-hard when it is three, or that we just lacked the necessary flash of insight to find a polynomial-time algorithm.
We present a constant-factor approximation algorithm for the general case with any number of machines in Section~\ref{sec:approx} and in Section~\ref{sec:PTAS} we give a polynomial-time approximation scheme for the case of a fixed number of machines. We leave the complexity of the problem (even for only three machines) as a tantalising open problem for the scheduling research community. We show in Section~\ref{sec:NPhard} that introducing weights for the jobs makes the problem NP-hard, already on 2 machines. We finish the paper by giving a table with the state of the art for other split scheduling problems with uniform setup times. We summarize whether they are known to be NP-hard or in P, and present the best known approximation ratios.

\section{Preliminaries}
\label{sec:P2}
An instance is given by $m$ parallel identical machines and $n$ jobs. Job $j$ has processing time $p_j$, for $j=1,\ldots,n$. Each job may be split into {\em parts} and multiple parts of the same job may be processed simultaneously. Before a machine can start processing a part of a job, a fixed setup time $s$ is required. During setup of a job (part) the machine cannot simultaneously process or setup another job (part). The objective is to minimise the sum of the completion times of the jobs (total completion time), which is equivalent to minimising the average completion time.

Here we derive some properties of an optimal schedule, which are valid for any number of machines. Some additional properties for the special case of two machines, presented in Section~\ref{sec:prop2}, will lead us to a polynomial-time algorithm for this special case. We show in Section~\ref{sec:3machines} that the additional properties that make the 2-machine case tractable do not hold for the case of three machines.

\label{sec:propm}
\begin{lemma}
\label{lem:1part}
Let $\sigma$ be a feasible schedule with job completion times $C_1\le C_2\le \dots\le C_n$. Let $\sigma'$ be obtained from $\sigma$ by rescheduling the job parts on each machine in order $1,2,\dots,n$. Then $C'_j\le C_j$ for $j=1,\dots,n$.
\end{lemma}
\begin{proof} Let $q_{ij}$ be the time that $j$ is processed on machine $i$ in $\sigma$ and let  $C_{ij}$ be the time that $j$ finishes on machine $i$. Let $y_{ij}=s+q_{ij}$ if $q_{ij}>0$ an let $y_{ij}=0$ otherwise. Fix some job $j$ and machine $i$. Let $k=\argmax\{C_{ik}\mid 1\le k\le j\}$.
Then $C_j\ge C_k\ge C_{ik}\ge \sum_{h=1}^j y_{ih}= C'_{ij}$, where the first inequality is by assumption and the last one by the fact that all work on jobs smaller than or equal to $j$ has been done on machine $i$ at time $C_{ik}$. Since  $C_j\ge C'_{ij}$ for any machine $i$ on which $j$ is scheduled the proof follows.
\end{proof}
The lemma above has several nice corollaries. First, note that if in an optimal schedule $C_1\le C_2\le \dots\le C_n$, then we maintain an optimal schedule with the same completion time for each job by scheduling the job parts on each machine in the order $1,2,\dots,n$. This allows to characterize an optimal schedule by a permutation of the jobs and the times that job $j$ is processed on each machine $i$. The optimal schedule is then obtained by adding a setup time $s$ for each non-zero job part and processing them in the order of the permutation on each machine. Consequently, in the optimal schedule obtained each machine contains at most one part of each job.

In the sequel, given a schedule, we use $M_j$ to denote the set of machines on which parts of job $j$ are processed.
We will sometimes say that a machine processes job $j$, if it processes a part of job $j$.

\begin{lemma}
\label{lem:SPT}
There exists an optimal schedule such that on each machine the job parts are processed (started and completed) in SPT order of the corresponding jobs.
\end{lemma}
\begin{proof}
We consider an optimal schedule of the form of Lemma 1, and among such schedules, we choose the schedule that minimizes $\sum_j p_jC_j$.

We claim that for any two jobs $j,k$ with $p_j<p_k$ and $C_j>C_k$ that are processed on the same machine, the machine must process $j$ before $k$.
In this case $M_j\cap M_k \neq \emptyset$. We define a new schedule by rescheduling both jobs within the time slots these jobs occupy in the current schedule (including the slots for the setup times).
First remove both jobs. Then consider the machines in $M_k$ one by one, starting with the machines in $M_k\backslash M_j$ and fill up the slot previously used by job $k$, until we have completely scheduled job $j$ including the setup times. This is possible since $p_j<p_k$. We will show that job $k$ can be scheduled in the remaining slots.

Let $M'_j$ and $M'_k$ denote the set of machines occupied by $j$ and $k$, respectively, in the new schedule. We distinguish two cases. If job $j$ cannot be rescheduled completely in the slots used by $k$ in $M_k\setminus M_j$ then we have $M'_k\subseteq M_j$. Together with $M'_j\subseteq M_k$ it follows that $(M'_j\cap M'_k)\subseteq (M_j\cap M_k)$. Hence, any machine containing both $j$ and $k$ in the new schedule did also contain both jobs in the old schedule. Hence, there are no extra setups on any machine needed.

Now consider the case that job $j$ is rescheduled completely in the slots used by $k$ in $M_k\setminus M_j$. Then, after adding job $k$, the total number of setups needed for $j$ and $k$ does not increase since there is at most one machine of $M_k\setminus M_j$ containing both jobs in the schedule, but none of the machines in $M_j \cap M_k$ is used in the new schedule.

Let $C'$ denote the new completion times. We have $C'_j \leq C_k$ and $C'_k \le \max\{C_j, C_k\}$, since in the new schedule $j$ is processed only where job $k$ was processed in the old schedule, and job $k$ is processed in the new schedule only where either job $j$ or job $k$ was processed in the old schedule.
For all other jobs, the completion time remains the same. Now, by assumption, we have that $C_j>C_k$, and hence $C'_k\le C_j$. Therefore, the sum of completion times did not increase, and $\sum_\ell p_\ell C'_\ell < \sum_\ell p_\ell C_\ell$, which contradicts the choice of the original schedule.

Hence, if there exist jobs $j,k$ such that $p_j<p_k$ and there exists some machine $i$ that processes job $k$ before $j$, then it must be the case that $C_j\le C_k$. If such jobs $j,k$ exist, then there also exist such jobs $j,k$ and a machine $i$ such that $k$ is processed immediately before $j$ on machine $i$. Now, reversing these parts in the schedule does not increase the completion of jobs $j$ and $k$ and does not effect the schedule of the other jobs. Moreover, we decrease the number of triples $i,j,k$ such that $p_j<p_k$ and machine $i$ processes $k$ before $j$.
Repeating this procedure gives an optimal schedule that satisfies the lemma.
\end{proof}

From now on, assume that jobs are numbered in SPT order, i.e., $p_1\le \dots \le p_n$.
Given a schedule, we call a job {\em balanced} if it completes at the same time on all machines on which it is processed.
\begin{lemma}
\label{lem:balanced}
There exists an optimal schedule in which all jobs are balanced.
\end{lemma}
\begin{proof}
Consider an optimal schedule with a minimum number of job parts. Let $C_j^*$ be the completion time of $j$ in this schedule and define $M_j$ for this schedule as before.
Consider the following linear program in which there is a variable $x_{ij}$ for all pairs $i,j$ with $i\in M_j$, indicating the amount of processing time of job $j$ assigned to machine $i$:
\begin{equation}
\label{e:LP}
\begin{array}{cll}
\mbox{min} & \displaystyle \sum_j C_j \\
\mbox{s.t.}& \displaystyle  \sum_{i\in M_j} x_{ij} = p_j, & \forall j=1, \ldots, n\\
& \displaystyle  \sum_{k\le j:\ i\in M_k} (s+x_{ik}) \le C_j, & \forall j=1, \ldots, n,\ \forall i\in M_j \\
&  x_{ij}\ge 0, C_j\ge 0, & \forall j=1, \ldots, n,\ \forall i\in M_j.
\end{array}
\end{equation}
Note that a schedule that satisfies Lemmas 1 and 2 gives a feasible solution to the LP, and on the other hand that any feasible solution to the LP gives a schedule with total completion time at most the objective value of the LP: if there exist some $j$ and $i\in M_j$ such that $x_{ij}=0$, then the LP objective value is at least the total completion time of the corresponding schedule, as there is no need to set up for job $j$ on machine $i$ if $x_{ij}=0$.
We know that a solution is a basic solution to this LP, only if the number of variables that are non-zero is at most the number of linearly independent tight constraints (not including the non-negativity constraints). By the minimality assumption on the optimal schedule, in any optimal solution to the LP all $C_j$ and $x_{ij}$ variables are non-zero, which gives a total of $n+\sum_j |M_j|$ variables. Since there are only $n+\sum_j |M_j|$ constraints, all constraints must be tight, which proves the lemma.
\end{proof}

\section{An $O(n\log n)$ time algorithm for two machines}
\label{sec:prop2}
Given a feasible schedule, we call a job $j$ a $d$-job, if $|M_j|=d$. In this section we assume that the number of machines is two.
\begin{lemma}\label{lem:1jobs}
Let $\sigma$ be an optimal schedule for a 2-machine instance and let $j<k$ be two consecutive 2-jobs. If there are 1-jobs between $j$ and $k$, then there is at least one 1-job on each machine. Also, the last 2-job is either not followed by any job or is followed by at least one 1-job on each machine.
\end{lemma}
\begin{proof}
Let $j$ and $k$ be two consecutive 2-jobs and assume there is at least one in-between 1-job on machine $1$ and none on machine 2.  Let $s_1,s_2$ be the start time of  job $j$ on respectively machine 1 and 2. We may assume without loss of generality that $s_1\ge s_2$: otherwise we just swap the schedules of the two machines for the interval $[0,C_j]$ and get the inequality. We change the schedule of $j$ and $k$ and the in-between 1-jobs as follows. Job $j$ is completely processed on machine 2, starting from time $s_2$,  and the in-between 1-jobs are moved forward such that the first starts at time $s_1$. Since $p_j\le p_k$, and the part of job $j$ we moved to machine 2 is at most $\frac{1}{2}p_j$, whereas the part of job $k$ that was on machine $j$ is at least $\frac{1}{2}p_k$, we can still schedule job $k$ such that its completion time remains the same. The total completion time is reduced by at least $s$ since $j$ needs only one setup time now. If $j$ is the last 2-job then we can make the same adjustment.
\end{proof}

\begin{lemma}
In the case of two machines there are no 1-jobs after a 2-job in an optimal schedule satisfying the properties of Lemmas~\ref{lem:1part}, \ref{lem:SPT} and \ref{lem:balanced}.
\end{lemma}
\begin{proof}
Suppose the lemma is not true.  Then there must be a 2-job $j$ that is directly followed by a 1-job. By Lemma \ref{lem:1jobs}, there must be at least one such 1-job on each machine, say jobs $h$ and $k$. Assume without loss of generality that $p_h \le p_k$. Let $x_{1j}, x_{2j}$ be the processing time of $j$ on machine 1 and 2, respectively.  As argued before, without loss of generality we assume that $x_{j1} \ge x_{j2}$. Let us define the starting time of $j$ as zero, and let $\Delta = x_{1j} - x_{2j}$.
Note that $C_j = \frac 12 (\Delta + p_j+2s)$.
Then, the sum of the three completion times is
\begin{eqnarray}
\label{eq:beforeswitch}
C_j + C_h+C_k &=& C_j + (C_j + p_h+s) + C_k  \nonumber \\
              &=& \Delta + p_j + 2s + p_h + s + C_k.
\end{eqnarray}
We reschedule the jobs $j,h,k$ as follows, the remaining schedule stays the same. Place
job $j$, the shortest among $j,k,h$, on machine 1 (unsplit), job $h$ on machine 2 (unsplit), and behind these two, job $k$ is split on machine 1 and 2, in such a way that it completes on one machine at time $C_h$ and time $C_k$ on the other.
The sum of the completion times of the three jobs becomes
\[(p_j +s) + (\Delta+ p_h +s) + C_k,\]
which is exactly $s$ less than the sum of the three completion times before the switch in (\ref{eq:beforeswitch}).
\end{proof}

\label{sec:polytime}
Given the previous lemmas, we see that the 2-jobs are scheduled in SPT order at the end.
By Lemma \ref{lem:SPT}, the first 2-job, say job $k$, is not shorter than the preceding 1-jobs. But this implies that the 1-jobs can be scheduled in SPT order
without increasing the completion time of job $k$ and the following jobs.
By considering each of the $n$ jobs as the first 2-job, we immediately obtain a $O(n^2)$ time algorithm to solve the problem. Carefully updating consecutive solutions leads to a faster method.

\begin{theorem}
There exists an $O(n\log n)$ algorithm for minimising the total completion time of jobs on two identical parallel machines with job splitting and uniform setup times.
\end{theorem}
\begin{proof}
Suppose we schedule
the first $k$ jobs (for any $1\leq k \leq n$) in SPT order as 1-jobs and the other jobs in SPT order as 2-jobs.
We would like to compute the change in objective value that results from changing job $k$ from a 1-job to a 2-job. However, this happens to give a rather complicated formula. It is much easier to consider the change for job $k-1$ and $k$ simultaneously.

The schedule for the 1-jobs $j < k - 1$ does not change.
To facilitate the exposition, suppose that job $k-1$ starts at time
zero and job $k$ starts at time $a$. Then $C_{k-1}+C_k = p_{k-1}+s+a+p_k +s$. After changing the jobs to 2-jobs, the new completion
times become $C'_{k-1} = (a + p_{k-1} + 2s)/2$ and $C'_k = (a + p_{k-1} + p_k + 4s)/2$.
Hence,
\[ C'_{k-1} + C'_k - C_{k-1} - C_k = s- p_k/2. \]
In addition, each job $j>k$ completes $s$ time units later. Hence, the total increase in objective value due to changing both job $k-1$ and $k$ from a 1-job to a 2-job is
\[ f(k):=(n-k+1)s-p_k/2.\]
Notice that $f(k)$ is decreasing in $k$, since $s>0$ and $p_k$ is nondecreasing in $k$.
Hence, either there exists some $k\in \{2, \ldots, n\}$ such that $f(k)< 0$ and $f(k-1)\ge 0$, or either $f(n) \ge 0$, or $f(2)<0$.

Suppose there exists some $k\in \{2, \ldots, n\}$ such that $f(k)< 0$ and $f(k-1)\ge 0$. The optimal schedule is to have either $k-1$ or $k-2$ unsplit jobs, since the first inequality and monotonicity implies that a schedule with $k-2$ unsplit jobs has a better objective value than a schedule with $k$ or more unsplit jobs, and the second inequality and monotonicity implies that a schedule with $k-1$ unsplit jobs has a better objective value than a schedule with $k-3$ or fewer unsplit jobs.

If $f(n) \ge 0$ then the optimal solution is either to have only 1-jobs or have only job $n$ as
a 2-job.
If $f(2)< 0$ then the optimal solution is either to have only 2-jobs or have only
job 1 as a 1-job.

Straightforward implementation of the above gives the desired algorithm, the running time of which is dominated by sorting the jobs in SPT order.
\end{proof}

\section{Troubles on more machines}
\label{sec:3machines}
The properties exposed in Section~\ref{sec:propm} have been proved to hold for any number of machines. The properties presented in Section~\ref{sec:prop2} are specific for two machines only. In this section we investigate their analogues for three and more machines. We will present some examples of instances that show that the extension is far from trivial. It keeps the complexity of the problem on three and more machines as an intriguing open problem.

Consider the instance on three machines having 10 jobs with their vector of processing times $p = (3,10,10,10,$ $10,50,50,50,50,50)$ (1 small job, 4 middle size jobs and 5 large jobs) and $s=0.7$. An optimal solution is depicted in Figure~\subref*{fig:inst1}. As we see, job 2 is split over machines 2 and 3, but job 3 starting later than job 2 is not split. Jobs 4 and 5 are again what we call 2-jobs and are split over machines 2 and 3. The large jobs are all split over all three machines.

We will consider two solutions to be the same, if one solution can be obtained from the other by a relabelling of machines, and/or (repeatedly) swapping the schedule of two machines from some time $t$ till the end of the schedule, if these two machines both complete processing of some job at time $t$. Also we will assume without loss of generality, that all processing times are distinct, by slightly perturbing the processing times if necessary, obtaining $p_j < p_{j+1}$ for all $j = 1,2,\ldots,n-1$.

The solution depicted in Figure~\subref*{fig:inst1} is not the unique optimal solution for this instance. The Gantt charts of the other three optimal solutions are given in the other three subfigures of Figure~\ref{fig:allopt}.  However, all optimal solutions for this instance share the property that job 2 is a 2-job, and either job 3 or job 4 is a 1-job. From this example we see that an optimal schedule does not necessarily have the property that $| M_j |$ is monotone in $j$.

\begin{figure*}[h!t]
\begin{center}
\subfloat[]{\label{fig:inst1}\noindent\resizebox{.7\textwidth}{!}{
\begin{ganttchart}[hgrid, inline, bar height = 1, bar top shift = 0, x unit = .13cm, y unit chart= .5cm]{102}
\input{GanttChart1}
\end{ganttchart}
}} \\
\subfloat[]{\label{fig:opt2}\noindent\resizebox{.7\textwidth}{!}{
\begin{ganttchart}[hgrid, inline, bar height = 1, bar top shift = 0, x unit = .13cm, y unit chart= .5cm]{102}
\input{GanttChart1a}
\end{ganttchart}
}}\\
\subfloat[]{\label{fig:opt3}\noindent\resizebox{.7\textwidth}{!}{
\begin{ganttchart}[hgrid, inline, bar height = 1, bar top shift = 0, x unit = .13cm, y unit chart= .5cm]{102}
\input{GanttChart2}
\end{ganttchart}
}}\\
\subfloat[]{\label{fig:opt4}\noindent\resizebox{.7\textwidth}{!}{
\begin{ganttchart}[hgrid, inline, bar height = 1, bar top shift = 0, x unit = .13cm, y unit chart= .5cm]{102}
\input{GanttChart2a}
\end{ganttchart}
}}\\
\end{center}
\caption{Gantt charts depicting the optimal solutions to the 3-machine instance with processing times $p = (3,10,10,10,10,50,50,50,50,50)$ (1 small job, 4 middle size jobs and 5 large jobs) and $s=0.7$. The grey blocks indicate the setup times, the numbered blocks are scheduled job parts. Each row of blocks gives the schedule for a machine. } \label{fig:allopt}
\end{figure*}

\begin{figure*}[h!t]
\begin{center}
\noindent\resizebox{.7\textwidth}{!}{
\begin{ganttchart}[hgrid, inline, bar height = 1, bar top shift = 0, x unit = .13cm, y unit chart= .5cm]{102}
\input{GanttChart3}
\end{ganttchart}
}
\end{center}
\caption{Gantt chart depicting the unique optimal solutions to the 3-machine instance with processing times $p = (3,10,10,10,10,50,50,50,50)$ (1 small job, 4 middle size jobs and 4 large jobs) and $s=0.7$. }\label{fig:inst2}
\end{figure*}

We now describe the other three optimal solutions for this instance.
The second optimal schedule, in Figure~\subref*{fig:opt3}, is obtained by scheduling job 1 on machine 1, job 2 split on machine 2 and 3, job 3 on machine 2 (or 3), and jobs 4 and 5 as split jobs on the machines not used by job 3. The remaining jobs are again all split on all three machines. It is easily verified that the objective of this schedule is the same as the objective of the schedule in Figure~\subref*{fig:inst1}: the completion time of job 3 increases by 2, and the completion times of jobs 4 and 5 each decrease by 1, and all other completion times remain the same.
The remaining two optimal schedules, in Figures~\subref*{fig:opt2} and \subref*{fig:opt4} are obtained by switching jobs 3 and 4. We note that this continues to hold if their processing times differ by some small value.

If we slightly change the instance by deleting one of the large jobs, then there is a unique optimal solution, which splits job 3 over machines 1 and 2 and continues with splitting job 4 over machines 1 and 3. Job 5 and the four large jobs are split over all three machines, see Figure~\ref{fig:inst2}.

That such a subtle change causes such a substantial change in the optimal schedule bodes ill for an algorithmic approach like the one in Section~\ref{sec:prop2}. These examples do not rule out that there exists an optimal schedule in which the jobs are started (or finished) in SPT order. A proof of either of these properties has proved elusive, however.

\section{Approximation algorithm}
\label{sec:approx}

We will now show a constant-factor approximation algorithm for our problem, for an arbitrary number of machines.
We remark that we do not know whether this problem is NP-hard, but the examples in the previous section do show that the way a job is scheduled in an optimal schedule may depend on jobs that occur later in the schedule. Our approximation algorithm, on the other hand, is remarkably simple, and only uses a job's processing time and the setup time to determine how to schedule the job.

We schedule the jobs in order of non-decreasing processing time. Let $s>0$ and let $\alpha$ be some constant that will be determined later. Job $j$ will be scheduled such that it completes as early as possible under the restriction that it uses at most $\ell_j := \min\{ \lceil \alpha p_j / s \rceil, m \}$ machines. Thus, the job will be scheduled on the at most $\ell_j$ machines that have minimum load in the schedule so far. It is easy to see that a job is always balanced this way.

\begin{theorem}
The algorithm described above is a $(2 + \alpha)$-approximation algorithm for minimising the total completion time with job splitting and uniform setup times, provided that $\alpha \geq \frac 14 (\sqrt{17}-1)$.
\end{theorem}

\begin{proof}
Let $\sigma$ be the schedule produced by the described algorithm. Note that the total load (processing times plus setup times) of all jobs in $\sigma$ up to, but not including, job $j$ is upper bounded by $L_j = \sum_{k < j} ( p_k + \ell_k s )$, since job $k$ introduced at most $\ell_k$ setups. Therefore, the average load on the $l_j$ least loaded machines is upper bounded by $L_j / m$. Since job $j$ is balanced, we can thus upper bound the completion time $\tilde C_j$ of job $j$ in the schedule by $L_j / m + p_j / \ell_j + s$. Note that this is an upper bound on the completion time of job $j$ when we try to schedule it on at most $\ell_j$ machines.

Noting that
\begin{eqnarray*}
p_j / \ell_j &=& p_j / \min\{ \lceil \alpha p_j / s \rceil, m \} \\
             &\leq &  p_j / \lceil \alpha p_j / s \rceil + p_j / m \\
             &\leq & ( 1/\alpha ) s + p_j / m,
\end{eqnarray*}
and
\[
	\ell_k s = \min\{ \lceil \alpha p_k / s \rceil, m \} s < \alpha p_k + s,
\]
we obtain
\begin{eqnarray*}
	\tilde C_j &\leq& L_j / m + p_j / \ell_j + s \\
	&\leq& \frac{1}{m}\sum_{k < j}  \left( p_k + \ell_k s  \right) + p_j / \ell_j + s \\
	&<& \frac{1}{m}\sum_{k < j}\big( ( 1 + \alpha ) p_k + s \big) + p_j / m + ( 1 + 1 / \alpha ) s \\
	&\leq& \frac{1+\alpha}{m} \sum_{k \leq j} p_k + \left( \frac{j-1}{m} + 1 + \frac{1}{\alpha} \right) s.
\end{eqnarray*}

We can lower bound the sum of completion times in an optimal schedule by $\sum_j (s + \frac 1m \sum_{ k\leq j } p_k)$:
suppose we only needed a setup time for the first job to be processed on a machine, for any machine. Clearly, the optimal sum of completion times for this problem gives a lower bound on the optimum.
Now, the optimal schedule when we only need a setup time for the first job on a machine processes the jobs in SPT order and splits each job over all machines, which gives a sum of completion times of $\sum_j (s+\frac 1m\sum_{k\le j} p_k)$.

Also, in any schedule, at most $m$ jobs are preceded by only one setup, at most another $m$ by two setups, etc., giving a lower bound of $\sum_j \lceil j / m \rceil s$ on the sum of completion times: this is exactly the optimal value when all processing times are $0$. We will show below that $\sum_j \lceil \frac jm \rceil s \ge \sum_j \frac{j-1}m s + \frac 12 ns$.

Hence, by using $1+\alpha$ times the first bound, and $1$ time the second bound, we get
\begin{eqnarray*}
\begin{array}{rll}
\left(2 + \alpha \right) & \sum_j C_j \geq
( 1+\alpha )\sum_j & \Big(s + \frac 1m \sum_{ k\leq j } p_k\Big) \\
 & & + \Big( \sum_j \frac{j-1}m s + \frac 12 ns \Big) \\
= &
\sum_j \left( \frac{ 1+\alpha  }{ m } \sum_{ k\leq j } p_k + \right. & \left.(1+\alpha)s + \frac{j-1}m s + \frac 12 s \right),
\end{array}
\end{eqnarray*}
which is at least as large as $\sum_j \tilde C_j$ provided $\alpha > 0$ and $\frac 32 + \alpha \geq 1 + \frac 1\alpha$, which is equivalent to $\alpha \geq \frac 14 (\sqrt{17}-1)$.

Next we show that $\sum_j \left\lceil \frac jm \right\rceil s \ge \sum_j \frac{j-1}m s + \frac 12 ns$. Let  $j=qm+a$ for some $q\ge 0$ and $a\in\{1,\dots,m\}$. Then
\begin{eqnarray*}
\left\lceil\frac jm \right\rceil-\frac{j-1}{m} &=& (q+1)-(qm+a-1)/m \\
                                               &=& 1-(a-1)/m.
\end{eqnarray*}
Now assume that $n = rm + b$, for some integer $r\ge 0$ and $b\in \{1,\ldots, m\}$. Then
\begin{eqnarray*}
\sum_{j=1}^{n}\left\lceil\frac jm \right\rceil-\sum_{j=1}^{n}\frac{j-1}{m}&=&r\sum_{a=1}^{m}\left(1-\frac{a-1}{m}\right)+\sum_{a=1}^{b}\left(1-\frac{a-1}{m}\right)\\
&=&rm+b-r\sum_{a=1}^{m}\frac{a-1}{m}-\sum_{a=1}^{b}\frac{a-1}{m}\\
&=&n-r(m-1)/2-\frac{1}{2}(b-1)b/m\\
&\ge &n-r(m-1)/2-\frac{1}{2}(b-1)\\
&=&n-(rm+b)/2+r/2+1/2\\
&=&n/2+r/2+1/2\ge n/2.\\
\end{eqnarray*}
Hence multiplying both sides with $s$ yields
\begin{eqnarray*}
\sum_{j=1}^{n}\left\lceil\frac jm \right\rceil s \geq \sum_{j=1}^{n}\frac{j-1}{m}s + \frac 12 ns.
\end{eqnarray*}
\end{proof}

\begin{corollary}
There exists a $2 + \frac 14 (\sqrt 17-1) < 2.781 $-approximation algorithm for minimising total completion time with job splitting and uniform setup times.
\end{corollary}

\section{A polynomial-time approximation scheme}\label{sec:PTAS}
We give an approximation scheme which runs in polynomial time if the number of machines is assumed constant. The idea is simple: by splitting a job $j$, at most $p_j$ on its completion time can be saved. It is easy to show that the value of a non-preemptive SPT schedule is no more than $\sum_j p_j$ larger than $\opt$. In particular, if we schedule the first $K=n-m/\epsilon$ jobs by non-preemptive SPT then the extra cost is at most $\sum_{j=1}^K p_j$. But, as we will see, this is only an $\epsilon$-fraction of the total completion time of the last $m/\epsilon$ jobs. These last jobs we schedule optimally given the schedule of the first $K$ jobs.

Now, we define the algorithm and its running time in more detail. Let, as before, $p_1\le \dots \le p_n$. Let $K=n-m/\epsilon$. (If $K\le 0$ then $n\le m/\epsilon$ and the optimal solution can be found in constant time.) Assume that $K$ is integer.
Let $\rho$ be an optimal schedule and let $\rho(K)$ be the schedule $\rho$ restricted to the jobs $1,2,\dots,K$. By Lemma \ref{lem:SPT} we may assume that $\rho(K)$ has no idle time. Let $t_i(\rho)$ be the completion time of machine $i$ in $\rho(K)$. The algorithm makes an approximate guess about the values $t_i(\rho)$. That means, it finds values $t_i$ such that
\begin{equation}\label{eq:t_i}
t_i(\rho)\le t_i\le t_i(\rho)+s+p_K.
\end{equation}
Note that for any $i$, we have $t_i(\rho)\le K(s+p_K)$. Hence, we need to try only $K^m$ guesses for $(t_1,\dots,t_m)$. Assume from now that we guessed $(t_1,\dots,t_m)$ correctly, i.e.,~(\ref{eq:t_i}) is satisfied.

We apply SPT to the jobs $1,2,\dots,K$ such that no machine $i$ is loaded more than $t_i+s+p_K$. This can easily be done as follows: apply list scheduling in SPT order and \emph{close} a machine once its load becomes $t_i$ or more. Let $T_i$ be the completion time of machine $i$ in the resulting schedule. Then $T_i\le t_i+s+p_K\le t_i(\rho)+2(s+p_K)$. Next, we find a near-optimal completion of the schedule by guessing for each job $j>K$ a set $M_j$ and apply linear programming. There are $2^{m(n-K)}$ possibilities for choosing such sets, which is a constant. The linear program works as follows. Note that the LP of Section~\ref{sec:P2} can be extended to do the following. Given a set $M_j$ for each job $j$ and a time $T_i$ for each machine, we can find the optimal schedule among all schedules for which: (i) job parts are in SPT order on each machine, (ii) machine $i$ does not start before $T_i$, (iii) job $j$ can only be scheduled on machines in $M_j$, and (iv) job $j$ has a setup time $s$ for each machine in $M_j$ even when its processing time $x_{ij}$ is zero. Note that it is not clear if the LP gives us the real optimal completion since we have not proved that the SPT properties hold also for optimal schedules if an initial part is fixed, as we do here. However, we can show that the solution given by the LP is close to optimal.

\paragraph{Approximation ratio}
Let $\sigma$ be the final schedule and let $\tilde C_j$ be the completion time of job $j$. Here we use $\opt$ to denote the objective value of optimal schedule $\rho$. For any $h\in\{1,\dots,n\}$ define $\mu_h= \sum_{k=1}^{h}(s+p_k)/m$. Then for any schedule, the $h$-th completion time is at least $\mu_h$.
Hence,
\begin{eqnarray*}\opt &\ge &\sum_{h=1}^n\mu_h \ge\sum_{h=K+1}^n\mu_h \ge  \sum_{h=K+1}^n\mu_K = (m/\epsilon)\mu_K \\ &=& \frac{1}{\epsilon}\sum_{k=1}^{K}(s+p_k).
\end{eqnarray*}
Further on in the proof we will use $C_h$ as the completion time of job $h$ in the optimal schedule. Here we use the notation $C^{(h)}$ for the $h$-th completion time of $\rho$, ($h=1,\dots,n$). Notice, that $C^{(h)}$ is not necessarily equal to $C_h$, the optimal completion time of job $h$ in $\rho$. For $h\le K$ it is easy to see that $\tilde C_h \le C^{(h)}+s+p_h$.
This implies
\begin{eqnarray}\label{eq:2}
\sum_{h=1}^{K} \tilde C_h &\le & \sum_{h=1}^{K}(C^{(h)} +s+p_h) \nonumber \\
                          & = & \sum_{h=1}^{K}C^{(h)}+\sum_{h=1}^{K}(s+p_h)  \\
                          &\le &  \sum_{h=1}^{K}C^{(h)} +\epsilon\opt. \nonumber
\end{eqnarray}
So for the first $K$ jobs we are doing fine. Next we give a bound on the total completion time of the other jobs.

Let $M^*_j$ be the set of machines used by job $j$ in the optimal schedule $\rho$. One of the guesses of the algorithm will be $M_j=M^*_j$ for $j>K$. We show that the corresponding LP-solution gives a near-optimal completion of the schedule.

A feasible LP solution is to take for $x_{ij}$, $j>K$, the values that correspond to $\rho$ and take choose values $C_j^{LP}=C_j+2(s+p_K)$, where we we remind that $C_j$ is the completion time of job $j$ in the optimal schedule $\rho$. The latter is feasible since $T_i\le t_i(\rho)+2(s+p_K)$. Hence, we can bound the total completion times of jobs $K+1,\ldots,n$ by
\begin{eqnarray}\label{eq:3}
 \sum_{h=K+1}^{n} \tilde C_h & \le & \sum_{h=K+1}^{n} C_h^{LP}  \nonumber \\
                             & \le & \sum_{h=K+1}^{n}(C_{h} +2(s+p_K)) \\
                             & =   & \sum_{h=K+1}^{n}C_{h} +2(n-K)(s+p_K). \nonumber
\end{eqnarray}
To bound the second term in the right hand side of (\ref{eq:3}) we derive another bound on $\opt$:
\begin{eqnarray*}
\opt & \ge & \sum_{h=1}^n\mu_h \ge\sum_{h=K+1}^n\mu_h \\
     & = & \sum_{h=K+1}^n \sum_{k=1}^{h}(s+p_k)/m \\
     & \ge & \sum_{h=K+1}^n \sum_{k=K+1}^{h}(s+p_k)/m \\
     & \ge & \sum_{h=K+1}^n (h-K)(s+p_K)/m  \\
     &  >  & \frac{1}{2}(n-K)^2(s+p_K)/m \\
     & = & \frac{1}{2} (n-K)(s+p_K)/\epsilon.
\end{eqnarray*}
Combining this with~(\ref{eq:3}) we get
\[
\sum_{h=K+1}^{n} \tilde C_h\le \sum_{h=K+1}^{n}C_{h} +   4\epsilon\opt \le \sum_{h=K+1}^{n}C^{(h)} +   4\epsilon\opt.
\]
Adding~(\ref{eq:2}) we can bound the total completion time by $(1+5\epsilon)\opt$.

\section{Hardness for weighted completion times}
We prove that introducing weights for the jobs in our problem makes it strongly NP-hard for any number of machines and weakly NP-hard for 2 machines.
\label{sec:NPhard}
\begin{theorem} The problem of minimising total weighted completion time with job splitting and uniform setup times on parallel identical machines ($P|\rm{s,split}|\sum w_jC_j$) is strongly NP-hard.
\end{theorem}
\begin{proof} We reduce from 3-\textsc{Partition}: given $3n$ positive numbers $a_1,\ldots,a_{3n}$ and a number $A$ such that $a_1+\dots+a_{3n}=nA$, does there exists a partition $A_1,\ldots,A_n$ of $\{1,\ldots,3n\}$ such that $|A_i|=3$ and $\sum_{j\in A_i}a_j=A$ for all $i$? Given an instance of 3-\textsc{Partition}, we construct the following instance of our scheduling problem: We have $n$ machines and $3n$ jobs. We set $p_j=a_j$ and $w_j=a_j+s$ for all $j=1,\ldots,3n$, where the setup time $s$ is some large enough number, to be defined later.

The idea behind the reduction is the following: the large setup time will make sure that exactly three jobs are scheduled (unsplit) per machine. The weights are chosen so that a schedule where all machines complete at exactly the same time is optimal, if such a schedule is feasible.

Suppose we schedule the jobs unsplit where $A_i$ is the set of jobs processed on machine $i$. Then, the cost of the schedule is:
\begin{eqnarray*}
\sum_{j=1}^{3n}w_jC_j&=&\sum_{i=1}^{n}\sum_{j\in A_i}w_jC_j\\
                     &=&\sum_{i=1}^{n}\sum_{j\in A_i}(s+a_j)\sum_{k\le j}(s+a_k)\\
										 &=&\sum_{i=1}^{n}\sum_{j\in A_i}\sum_{k\le j}(s+a_j)(s+a_k)\\
										 &=&\frac{1}{2}\sum_{i=1}^{n}\left[\left(\sum_{j\in A_i}(s+a_j)\right)^2+\sum_{j\in A_i}(s+a_j)^2\right]\\
										 &=&\frac{1}{2}\sum_{i=1}^{n}l_i^2+\frac{1}{2}\sum_{j=1}^{3n}(s+a_j)^2,
\end{eqnarray*}										
where $l_i$ is the total load on machine $i$. Note that the second term is independent of the schedule. This cost is minimised when $l_i=l_h$ for all $i,h$ and this can be realised if a perfect 3-partition exists. Let us denote this minimum by $\opt_{\rm{3P}}$.

If no perfect 3-partition exists, then any schedule where no jobs are split has strictly higher cost than $\opt_{\rm{3P}}$. It remains to prove that also any schedule with at least one split job has a strictly higher cost than $\opt_{\rm{3P}}$.

First observe that \begin{eqnarray*}
\opt_{\rm{3P}}&=& \frac{1}{2}\sum_{i=1}^{n}(3s+A)^2+\frac{1}{2}\sum_{j=1}^{3n}(s+a_j)^2 \\
              &=& 6ns^2+O(ns).
\end{eqnarray*}
Now assume that at least one job is split, then there are at least $3n+1$ setup times of $s$ each. Consider the extreme case where all $3n$ values $a_j$ are zero. In this case it is easy to see that the weighted sum of the $3n$ completion times is at least $(6n+1)s^2$. Clearly, this bound holds as well for arbitrary value $a_j$. For large enough $s$ we have $(6n+1)s^2>\opt_{\rm{3P}}$.
\end{proof}

\begin{theorem} The problem $P2|\rm{s,split}|\sum w_jC_j$ is weakly NP-hard.\end{theorem}
\begin{proof}
We now reduce from a restricted form of the \textsc{Subset Sum} problem: Given $2n$ positive integers $a_1,\dots,a_{2n}$ such that $a_1+\dots+a_{2n}=2A$, is there a set $I\subset \{1,\dots,2n\}$ such that $|I|=n$ and $\sum_{i\in I} a_i=A$?
Given an instance of \textsc{Subset Sum}, we construct the following instance of our scheduling problem. We have $2$ machines and $2n$ jobs. We set $p_j=a_j$ and $w_j=a_j+s$ for  $j=1,\ldots,2n$, where the setup time $s$ is some large enough number, to be defined later. The proof follows the same reasoning as the previous proof: the large setup time will now make sure that exactly $n$ jobs are scheduled (unsplit) per machine, and the weights will make sure that a schedule where the two machines complete at exactly the same time is optimal, if such a schedule is feasible.

Suppose we schedule the jobs unsplit. Then, just as in the proof above for an arbitrary number of machines we have that the cost of the schedule is:
\begin{eqnarray*}
\sum_{j=1}^{2n}w_jC_j=\frac{1}{2}(l_1^2+l_2^2)+\frac{1}{2}\sum_{j=1}^{2n}(s+a_j)^2,
\end{eqnarray*}										
where $l_i$ is the total load on machine $i$. Note that the second term is independent of the schedule. This cost is minimised when $l_1=l_2$ and this can be realised if a perfect subset $I$ exists. Let us denote this minimum by $\opt_{\rm{S}}$.

If no perfect subset exists, any unsplit schedule has strictly higher cost. It remains to prove that also any schedule with at least one split has a strictly higher cost than $\opt_{\rm{S}}$.

First observe that \begin{eqnarray*}
\opt_{\rm{S}} &=& (ns+A)^2+\frac{1}{2}\sum_{j=1}^{2n}(s+a_j)^2 \\
              &=& (n^2+n)s^2+O(ns).
\end{eqnarray*}
Now assume that at least one job is split, then there are at least $2n+1$ setup times of $s$ each. Consider the extreme case where all $2n$ values $a_j$ are zero. In this case it is easy to see that the weighted sum of the $2n$ completion times is at least $(n^2+n+1)s^2$. Clearly, this bound holds as well for arbitrary values $a_j$. For large enough $s$ we have $(n^2+n+1)s^2>\opt_{\rm{S}}$.
\end{proof}

\section{Epilogue}
\label{sec:epilogue}
In the following table we gather the state of the art on scheduling problems with job splitting and uniform setup times. For describing the problems in the first column of the table we use the standard three-field scheduling notation \cite{graham1979}. In the first field, expressing the processor environment, we only consider parallel identical machines, denoted by $P$, possibly with the number of parallel machines mentioned additionally. In the second field, expressing job characteristics, the term $pmtn$ denotes ordinary preemption, $split$ denotes job splitting as we consider in this paper and $s$ denotes the presence of uniform setup times.
Though this paper is mainly concerned with problems with a total completion time objective, indicated by $\sum C_j$ in the third field, expressing the objective, we will also show the state of the art on the total weighted completion time (indicated by $\sum w_j C_j$) and on the makespan (indicated by $C_{max}$).

In the second column, we summarize the complexity status of these problems. A question mark indicates that the complexity of the problem is unknown. In the third column we give the best approximation guarantee known, where a `-' indicates that no algorithm with a performance guarantee is known.
If we consider it relevant, we also present, as a footnote, the knowledge on the comparable version with preemption instead of splitting.

\begin{table*}
\caption{Minimising Total (Weighted) Completion Time and Makespan with Job Splitting}
\label{overview2}
\begin{center}
\begin{tabular}[t]{lll}
Problem & Complexity & Algorithm \\
\hline \hline \\
\multirow{2}{*}{$P\:|\:split\:|\:\sum{C_{j}}$} & \multirow{2}{*}{in P} & divide jobs equally
over the \\
&&machines in SPT order
\\\\
$P2\:|\:s, split\:|\:\sum{C_{j}}$ & in P & algorithm of Section~\ref{sec:polytime}
\\\\
$P\:|\:s, split\:|\:\sum{C_{j}}$ & ? & 2.781-approx. in Section~\ref{sec:approx}
\\
\scriptsize{cf. $P\:|\:s, pmtn\:|\:\sum{C_{j}}$} & \scriptsize{in P} & \scriptsize{SPT}
\\\\
\hline \\
\multirow{2}{*}{$P\:|\:split\:|\:\sum{w_{j}C_{j}}$} & \multirow{2}{*}{in P} &
divide jobs
equally over the \\
&&machines in WSPT order
\\
\scriptsize{cf. $P\:|\:pmtn\:|\:\sum{w_{j}C_{j}}$} & \scriptsize{NP-hard
\cite{brunocoffmansethi1974}} & \scriptsize{PTAS \cite{Afrati}}
\\\\
$P\:|\:s, split\:|\:\sum{w_{j}C_{j}}$ & NP-hard & -
\\
\scriptsize{cf. $P\:|\:s, pmtn\:|\:\sum{w_{j}C_{j}}$} & \scriptsize{NP-hard} & \scriptsize{-}
\\\\
\hline \\
$P\:|\:s, split\:|\:C_{max}$ & NP-hard \cite{ster2010} (cf. \cite{chen2006lot}) & $\frac{5}{3}$-approx. split/assignment \cite{chen2006lot}
\\\\
$P\:|\:s, split\:|\:C_{max}$ & NP-hard & $\frac{3}{2}$-approx. wrap-around
alg. \\
if $p_j \geq s$ $\forall j$ \\
\scriptsize{cf. $P\:|\:s, pmtn\:|\:C_{max}$} & \scriptsize{NP-hard \cite{schuurman1999preemptive}} & \scriptsize{PTAS \cite{schuurman1999preemptive}} \\\\
\hline \hline
\end{tabular}
\end{center}
\end{table*}

\bibliographystyle{plain}
\bibliography{SplitScheduling}

\begin{thebibliography}{10}

\bibitem{Afrati}
Foto Afrati, Evripidis Bampis, Chandra Chekuri, David Karger, Claire Kenyon,
  Sanjeev Khanna, Ioannis Milis, Maurice Queyranne, Martin Skutella, Cliff
  Stein, and Maxim Sviridenko.
\newblock Approximation schemes for minimizing average weighted completion time
  with release dates.
\newblock In {\em Proc. 40th Annual Symposium on Foundations of Computer
  Science}, pages 32--43, 1999.

\bibitem{brunocoffmansethi1974}
J.~Bruno, E.G. Coffman~Jr., and R.~Sethi.
\newblock Scheduling independent tasks to reduce mean finishing time.
\newblock {\em Communications of the ACM}, 17(7):pp. 382--387, 1974.

\bibitem{chen2006lot}
B.~Chen, Y.~Ye, and J.~Zhang.
\newblock Lot-sizing scheduling with batch setup times.
\newblock {\em Journal of Scheduling}, 9(3):299--310, 2006.

\bibitem{Drozdowski09}
Maciej Drozdowski.
\newblock {\em Scheduling for Parallel Processing}.
\newblock Springer, 1st edition, 2009.

\bibitem{DuLeungYoung}
Jianzhong Du, Joseph Y.-T. Leung, and Gilbert~H. Young.
\newblock Minimizing mean flow time with release time constraint.
\newblock {\em Theoretical Computer Science}, 75(3):374--355, 1990.

\bibitem{graham1979}
R.L. Graham, E.L. Lawler, J.K. Lenstra, and A.H.G.R. Kan.
\newblock Optimization and approximation in deterministic sequencing and
  scheduling: a survey.
\newblock {\em Annals of discrete mathematics}, 5(2):287--326, 1979.

\bibitem{lenstra88}
Jan~Karel Lenstra.
\newblock The mystical power of twoness: in memoriam {Eugene L.} {Lawler}.
\newblock {\em Journal of Scheduling}, 1(1):3--14, 1998.

\bibitem{LiuC04}
Zhaohui Liu and T.~C.~Edwin Cheng.
\newblock Minimizing total completion time subject to job release dates and
  preemption penalties.
\newblock {\em J. Scheduling}, 7(4):313--327, 2004.

\bibitem{PottsvW92}
C.~N. Potts and L.~N.~Van Wassenhove.
\newblock Integrating scheduling with batching and lot-sizing: A review of
  algorithms and complexity.
\newblock {\em The Journal of the Operational Research Society}, 43(5):pp.
  395--406, 1992.

\bibitem{schuurman1999preemptive}
P.~Schuurman and G.J. Woeginger.
\newblock Preemptive scheduling with job-dependent setup times.
\newblock In {\em Proceedings of the tenth annual ACM-SIAM symposium on
  Discrete algorithms}, pages 759--767. Society for Industrial and Applied
  Mathematics, 1999.

\bibitem{ster2010}
Suzanne {van der Ster}.
\newblock The allocation of scarce resources in disaster relief, 2010.
\newblock MSc-Thesis in Operations Research at VU University Amsterdam.

\bibitem{XingZ00}
Wenxun Xing and Jiawei Zhang.
\newblock Parallel machine scheduling with splitting jobs.
\newblock {\em Discrete Applied Mathematics}, 103(1-3):259--269, 2000.

\end{thebibliography}

\end{document}